\documentclass[runningheads]{llncs}
\usepackage[T1]{fontenc}

\usepackage{amsmath,amssymb,amsfonts}
\usepackage{algorithmicx, algorithm,algpseudocode}
\usepackage{longtable}
\usepackage{enumerate}
\usepackage{graphicx}
\usepackage{orcidlink}
\usepackage{algpseudocode}
\usepackage{algorithm}
\usepackage{stmaryrd}
\usepackage{float}
\usepackage{multirow}

\newtheorem{prp}{Proposition}

\newtheorem{exm}{Example}

\usepackage{enumerate}
\usepackage{hyperref}
\hypersetup{
	colorlinks   = true, %Colours links instead of ugly boxes
	urlcolor     = red, %Colour for external hyperlinks
	linkcolor    = blue, %Colour of internal links
	citecolor   = blue %Colour of citations, could be ``red''
}
\newcommand{\splitatcommas}[1]{%
	\begingroup
	\begingroup\lccode`~=`, \lowercase{\endgroup
		\edef~{\mathchar\the\mathcode`, \penalty0 \noexpand\hspace{0pt plus 1em}}%
	}\mathcode`,="8000 #1%
	\endgroup
}
\usepackage[
n,
operators,
advantage,
sets,
adversary,
landau,
probability,
notions,	
logic,
ff,
mm,
primitives,
events,
complexity,
asymptotics,
keys]{cryptocode}

\begin{document}
\title{A New Algorithm for Computing Branch Number of Non-Singular Matrices over Finite Fields}
\titlerunning{A New Algorithm for Computing Branch Number}

%\author{}
%\authorrunning{}
%\institute{}

\author{
P.R. Mishra \inst{1} \orcidlink{0000-0003-0473-6616}\and
Yogesh Kumar \inst{1} \orcidlink{0000-0001-8495-102X}\and
Susanta Samanta \inst{2} \orcidlink{0000-0003-4643-5117} \and
Atul Gaur \inst{3} \orcidlink{0000-0002-8723-4541}
}

\authorrunning{Mishra et al.}
\institute{
Scientific Analysis Group, DRDO, Metcalfe House Complex, Delhi-110054 \\ \email{\{prasanna.r.mishra,adhana.yogesh\}@gmail.com} 
\and R. C. Bose Centre for Cryptology and Security, Indian Statistical Institute, Kolkata-700108 \\ \email{susanta.math94@gmail.com} \and
Department of Mathematics, University of Delhi, Delhi-110007 \\ \email{gaursatul@gmail.com}
}
\maketitle              % typeset the header of the contribution

\setcounter{footnote}{0} % Reset footnote counter

\begin{abstract}
The notion of branch number of a linear transformation is crucial for both linear and differential cryptanalysis. The number of non-zero elements in a state difference or linear mask directly correlates with the active S-Boxes. The differential or linear branch number indicates the minimum number of active S-Boxes in two consecutive rounds of an SPN cipher, specifically for differential or linear cryptanalysis, respectively. This paper presents a new algorithm for computing the branch number of non-singular matrices over finite fields. The algorithm is based on the existing classical method but demonstrates improved computational complexity compared to its predecessor. We conduct a comparative study of the proposed algorithm and the classical approach, providing an analytical estimation of the algorithm's complexity. Our analysis reveals that the computational complexity of our algorithm is the square root of that of the classical approach.

\end{abstract}

\section{Introduction}
Claude Shannon, in his seminal paper ``Communication Theory of Secrecy Systems'' \cite{shan} introduced the concepts of confusion and diffusion, which are crucial properties for any cryptographic primitive. In cryptographic primitives, achieving the diffusion property is typically realized through the incorporation of a linear layer, often represented as a matrix. This matrix is designed to induce a significant alteration in the output for even minor changes in the input.  The strength of the diffusion layer is usually assessed by its branch number, and the optimal branch number is achieved by the use of MDS matrices.

In~\cite{JDA_Thesis_1995}, Daemen introduced the concept of the branch number of a linear transformation as a measure of its diffusion. The branch number can be categorized into two types: differential branch number and linear branch number. The differential (or linear) branch number of a matrix $M$ is defined as the smallest sum of the Hamming weight of the input vector $x$ and the output vector $Mx$ (or its transpose $M^T$). The notion of branch number for linear transformations is crucial for both linear and differential cryptanalysis. It provides the minimum number of active S-boxes we can expect in any valid differential (or linear) trail over two rounds. Specifically, for any two-round differential (or linear) trail, the number of active S-boxes is lower bound by the differential (or linear) branch number. Thus, a higher branch number implies better diffusion and therefore stronger resistance against differential cryptanalysis~\cite{biham1991differential} and linear cryptanalysis~\cite{matsui1993linear}.

The optimal branch number of MDS matrices makes them a preferred choice for designing diffusion layers in many block ciphers and hash functions. An example of such utilization is seen in the MixColumns operation of AES~\cite{AES}. Furthermore, MDS matrices are widely applied in various cryptographic primitives, including stream ciphers like MUGI~\cite{MUGI} and hash functions such as Maelstrom~\cite{MAELSTROM}, Gr$\phi$stl~\cite{GROSTL}, and PHOTON~\cite{PHOTON}. This widespread application underscores the effectiveness of MDS matrices within diffusion layers, which prompts the development of various techniques for their design. We refer to \cite{kcz} for various constructions on MDS matrices.

On the other hand, Near-MDS matrices with sub-optimal branch numbers provide a better balance between security and efficiency as compared to MDS matrices. Several lightweight block ciphers, such as PRIDE \cite{PRIDE}, Midori \cite{MIDORI}, MANTIS \cite{SKINNY}, FIDES \cite{Fides}, and PRINCE \cite{PRINCE}, have leveraged Near-MDS matrices. The importance of lightweight symmetric key primitives with attributes such as low power consumption, minimal energy usage, or reduced latency is growing, and Near-MDS matrices are commonly employed in their construction. However, Near-MDS matrices have received comparatively less attention in the existing literature. Some relevant studies include \cite{Gupta2023direct,NMDS_Gupta2024,Li_Wang_2017,NMDS_code_2022,NMDS_code_2022_2}.

We can determine whether a matrix of order $n$ is MDS or Near-MDS simply by checking its branch number. For MDS matrices, both differential and linear branch numbers are $n+1$, while for Near-MDS matrices, they are $n$. Therefore, one might want to explore specific techniques or algorithms for efficiently checking the branch number of a matrix. Moreover, it's worth noting that the differential branch number of a matrix $A$ is equal to the minimum distance of a linear code $\mathcal{C}$ generated by the matrix $[I~|~A]$. Additionally, the linear branch number is equivalent to the minimum distance of the dual code $\mathcal{C}^{\perp}$ of $\mathcal{C}$. Therefore, finding efficient techniques for calculating the branch number not only impacts the design and analysis of symmetric key cryptographic algorithms but also provides insights for linear error-correcting codes.

One method to determine the branch number of a matrix is by a straightforward approach, which involves finding the smallest sum of the Hamming weights of the input vector $x$ and the output vector $Mx$. In~\cite[Appendix A]{guo2016constructing}, Guo et al. noted that the branch number $d$ of a binary non-singular matrix $M$ of order $n$ can be identified by searching for the minimum value of $w_h(x) + w_h(M x)$ and $w_h(x) + w_h(M^{-1} x)$ among the input vectors with weights up to $d/2$, where $d \leq n+1$ and $w_h(x)$ represents the Hamming weight of the vector $x$. However, the paper lacks a proof or detailed explanation of this observation.

In this paper, we extend the concept introduced by Guo et al. to encompass non-singular matrices over any finite field $\mathbb{F}_q$. Additionally, we provide a mathematical foundation for our method of computing the branch number of a non-singular matrix $M$ of order $n$ over $\mathbb{F}_q$. Moreover, as an enhancement, we demonstrate that it is unnecessary to consider all input vectors with weights up to $d/2$. Instead, only a few select vectors are required. Specifically, we illustrate that we can compute the branch number by focusing on the equivalence classes. Therefore, we can determine the branch number by calculating over the representatives of these equivalence classes, thus improving the efficiency of our approach. Additionally, we analyze the time complexity of our proposed algorithm and identify specific instances of non-singular matrices where the complexity can be further reduced.

\noindent This paper is structured as follows: Section~\ref{sec:definition} covers the preliminary concepts necessary for understanding our work. We present our proposed algorithm for computing the branch number of non-singular matrices over a finite field $\mathbb{F}_q$ in Section~\ref{sec:proposed_algorithm}. In Section~\ref{sec:computational_complexity}, we delve into the computational complexity of our algorithm compared to the direct approach. Additionally, we demonstrate that our proposed algorithm requires nearly the square root of the field multiplications needed for the direct approach. Finally, we conclude the paper in Section~\ref{sec:conclusion}.

\section{Preliminaries}\label{sec:definition}
In this section, we discuss some definitions and mathematical preliminaries that are important in our context. 

Let $\mathbb{F}_{p^m}$ represent a finite field of order $p^m$, where $p$ is a prime and $m$ is a positive integer. We denote the multiplicative group of the finite field $\mathbb{F}_{p^m}$ by $\mathbb{F}_{p^m}^*$. The set of vectors of length $n$ with entries from the finite field $\mathbb{F}_{p^m}$ is denoted by $\mathbb{F}_{p^m}^n$. The Hamming weight of a vector $x=(x_1,\ldots,x_n)\in \mathbb{F}_{p^m}^n$, denoted by $w_h(x)$, is the total number of non-zero components in the vector $x$, i.e., $w_h(x)=|\{ i \in \{1, 2, \ldots, n\} : x_i \neq 0\}|$. 

An $n\times n$ matrix, or a matrix of order $n$ over $\mathbb{F}_{p^m}$, is considered non-singular if the determinant of the matrix is a non-zero element in $\mathbb{F}_{p^m}$. Now we define the notions of differential and linear branch numbers. 

\begin{definition}~\cite{JDA_Thesis_1995}\label{Def_diff_branch}
The differential branch number, $\mathcal B_d(M)$, of a matrix $M$ of order $n$ over the finite field $\mathbb{F}_{q}$ is defined as the smallest number of non-zero components in both the input vector $x$ and the output vector $Mx$, as we consider all non-zero $x$ in $\mathbb{F}_{q}^n$ i.e.

    \begin{equation*}
    	\begin{aligned}
                \mathcal B_d(M)=\min_{x\neq 0} \set{w_h(x)+w_h(Mx)}.
    	\end{aligned}
    \end{equation*}
\end{definition}

\begin{definition}~\cite{JDA_Thesis_1995,AES}\label{Def_lin_branch}
The linear branch number, $\mathcal B_l(M)$, of a matrix $M$ of order $n$ over the finite field $\mathbb{F}_{q}$ is defined as the smallest number of non-zero components in both the input vector $x$ and the output vector $M^Tx$, as we consider all non-zero $x$ in $\mathbb{F}_{q}^n$ i.e.
    \begin{equation*}
    	\begin{aligned}
    		\mathcal B_l(M)=\min_{x\neq 0} \set{w_h(x)+w_h(M^Tx)}.
    	\end{aligned}
    \end{equation*}
\end{definition}

\begin{remark}\cite[Page $144$]{AES}\label{Remark_min_dis_is_branch_number}
It is important to note that the differential branch number $\mathcal{B}_d(M)$ of a matrix $A$ equals the minimum distance of a linear code $\mathcal{C}$ generated by the matrix $[I~|~A]$. Furthermore, $\mathcal{B}_l(M)$ is equivalent to the minimum distance of the dual code $\mathcal{C}^{\perp}$ of $\mathcal{C}$.
\end{remark}

\begin{remark}\cite[Page $132$]{AES}\label{Remark_max_branch}
It is noteworthy that the maximum value for both $\mathcal{B}_d(M)$ and $\mathcal{B}_l(M)$ is $n + 1$. Although $\mathcal{B}_d(M)$ and $\mathcal{B}_l(M)$ may not always be equal, a matrix with the highest possible differential or linear branch number will have the same value for both.
\end{remark}

\noindent In this paper, our focus lies on the differential branch number of a matrix $M$, which we will simply refer to as the branch number and denote as $\mathcal{B}(M)$.

\section{The New Algorithm for Computation of Branch Number}  \label{sec:proposed_algorithm}
In this section, we propose a new algorithm for computing the branch number of a non-singular matrix over a finite field. The algorithm builds upon existing methods but demonstrates improved computational complexity by leveraging techniques from linear algebra and finite field arithmetic. The following theorem ensures that it is unnecessary to consider vectors of all weights in the computation of the branch number.

\begin{theorem}\label{Thm_The_Branch_Number_nonsingular_matrix}
The branch number of an invertible matrix $M\in M_n(\mathbb{F}_q)$ is given as
\[\begin{split}
\mathcal B(M)=\min\left\{\min\left\{h(M,x),h(M^{-1},x)\right\}\mid x\in \mathbb{F}_q^n, 1\leq w_h(x)\leq \left\lfloor \frac{n+1}{2}\right\rfloor\right\}, 
\end{split}\]
where $h(M,x)=w_h(x)+w_h(Mx)$.
\end{theorem} 

\begin{proof}
To begin with, recall that for an invertible matrix $M$  in $M_n(\mathbb{F}_q )$, where $n>1$, the branch number $\mathcal B(M)$ of $M$ is given as
\[ \mathcal B(M)=\min\{h(M,x)\mid x\in \mathbb{F}_q^n, x \neq 0\}.\]

As $ x\neq 0 \Rightarrow w_h(x)\neq 0$. Consequently, we may write 
\[\mathcal B(M)=\min\{h(M,x)\mid x\in \mathbb{F}_q^n, 1\leq w_h(x)\leq n\}.\]
	
\noindent We partition the set $\left\{1,\ldots,n\right\}$ in two parts viz. $\left\{1,\ldots,\left\lfloor \frac{n+1}{2}\right\rfloor \right\}$ and $\left\{\left\lfloor \frac{n+1}{2}\right\rfloor +1,\right.\\
\left.\ldots,n\right\}$ to compute $\mathcal B(M)$ as

\begin{equation}\label{Thm_branch_number_eqn:1}
    \begin{split}
        \mathcal B(M)=\min\left\{\min\left\{h(M,x)\mid x\in \mathbb{F}_q^n, 1\leq w_h(x)\leq \left\lfloor \frac{n+1}{2}\right\rfloor\right\},\right.\\
        \left.\min\left\{h(M,x)\mid x\in \mathbb{F}_q^n, \left\lfloor \frac{n+1}{2}\right\rfloor <w_h(x)\leq n\right\}\right\}. 
    \end{split}
\end{equation}

Next, we divide the second term on the right-hand side of (\ref{Thm_branch_number_eqn:1}) into the cases where $w_h(Mx)\leq \floor{\frac{n+1}{2}}$ and $w_h(Mx)> \floor{\frac{n+1}{2}}$. Therefore, we have

\begin{equation}\label{Thm_branch_number_eqn:2}
    \begin{split}
        &\min\left\{h(M,x)\mid x\in \mathbb{F}_q^n, \left\lfloor \frac{n+1}{2}\right\rfloor < w_h(x)\leq n\right\}\\
        =&\min\left\{\min\left\{h(M,x)\mid x\in \mathbb{F}_q^n, \left\lfloor \frac{n+1}{2}\right\rfloor<w_h(x)\leq n, w_h(Mx)\leq \left\lfloor \frac{n+1}{2}\right\rfloor \right\},\right.\\
        &\left.\min\left\{h(M,x)\mid x\in \mathbb{F}_q^n, \left\lfloor \frac{n+1}{2}\right\rfloor<w_h(x)\leq n, w_h(Mx)> \left\lfloor \frac{n+1}{2}\right\rfloor \right\}\right\}.
    \end{split}
\end{equation}

Note that for the second term of the right-hand side of Equation (\ref{Thm_branch_number_eqn:2}), $h(M,x)=w_h(x)+w_h(Mx)>2\left\lfloor \frac{n+1}{2}\right\rfloor+1 \geq n+1$. However, we know that the upper bound for $\mathcal B(M)$ is $n+1$. Thus, we conclude that the second term of the right-hand side of (\ref{Thm_branch_number_eqn:2}) will not contribute to the computation of the branch number.
\\

\noindent Therefore, from (\ref{Thm_branch_number_eqn:1}) and (\ref{Thm_branch_number_eqn:2}), we have
\begin{equation} \label{Thm_branch_number_eqn:3}
    \begin{split} 
        \mathcal B(M)=&\min\left\{\min\left\{h(M,x)\mid x\in \mathbb{F}_q^n, 1\leq w_h(x)\leq \left\lfloor \frac{n+1}{2}\right\rfloor\right\},\right.\\
        &\left.\min\left\{h(M,x)\mid x\in \mathbb{F}_q^n, \left\lfloor \frac{n+1}{2}\right\rfloor<w_h(x)\leq n, w_h(Mx)\leq \left\lfloor \frac{n+1}{2}\right\rfloor\right\}\right\}. 
    \end{split}
\end{equation}

\noindent Again, we note that
\begin{equation*}  
    \begin{split} 
        &\left\{h(M,x)\mid x\in \mathbb F_q^n, 1\leq w_h(x)\leq \left\lfloor \frac{n+1}{2}\right\rfloor, w_h(Mx)\leq \left\lfloor \frac{n+1}{2}\right\rfloor\right\} \subseteq\\
        &\left\{h(M,x)\mid x\in \mathbb F_q^n, 1\leq w_h(x)\leq \left\lfloor \frac{n+1}{2}\right\rfloor\right\}.
    \end{split}
\end{equation*}
Therefore, 
\begin{equation} \label{Thm_branch_number_eqn:4}
    \begin{split} 
        &\min \left\{h(M,x)\mid x\in \mathbb F_q^n, 1\leq w_h(x)\leq \left\lfloor \frac{n+1}{2}\right\rfloor\right\} \leq \\
        &\min \left\{h(M,x)\mid x\in \mathbb F_q^n, 1\leq w_h(x)\leq \left\lfloor \frac{n+1}{2}\right\rfloor, w_h(Mx)\leq \left\lfloor \frac{n+1}{2}\right\rfloor\right\}.
    \end{split}
\end{equation}
Note that the right-hand side of (\ref{Thm_branch_number_eqn:4}) is always greater than or equal to the left-hand side of (\ref{Thm_branch_number_eqn:4}). Therefore, if we include this extra term in (\ref{Thm_branch_number_eqn:3}), it will not affect the minimum value computed in (\ref{Thm_branch_number_eqn:3}).
\begin{equation} \label{Thm_branch_number_eqn:5}
    \begin{split}
        \mathcal B(M)=&\min\left\{\min\left\{h(M,x)\mid x\in \mathbb F_q^n, 1\leq w_h(x)\leq \left\lfloor \frac{n+1}{2}\right\rfloor\right\},\right.\\
        &\min\left\{h(M,x)\mid x\in \mathbb F_q^n, 1\leq w_h(x)\leq \left\lfloor \frac{n+1}{2}\right\rfloor, w_h(Mx)\leq \left\lfloor \frac{n+1}{2}\right\rfloor\right\},\\
        &\left.\min\left\{h(M,x)\mid x\in \mathbb F_q^n, \left\lfloor \frac{n+1}{2}\right\rfloor<w_h(x)\leq n, w_h(Mx)\leq \left\lfloor \frac{n+1}{2}\right\rfloor\right\}\right\}
    \end{split}
\end{equation}
By merging the second and third terms on the right-hand side of (\ref{Thm_branch_number_eqn:5}), we obtain 
\begin{equation*} 
    \begin{split}
        \mathcal B(M)=&\min\left\{\min\left\{h(M,x)\mid x\in \mathbb F_q^n, 1\leq w_h(x)\leq \left\lfloor \frac{n+1}{2}\right\rfloor\right\},\right.\\
        &\left.\min\left\{h(M,x)\mid x\in \mathbb F_q^n, 1\leq w_h(x)\leq n, w_h(Mx)\leq \left\lfloor \frac{n+1}{2}\right\rfloor\right\}\right\}. 
    \end{split}
\end{equation*}

Let $Mx=y$, then $x=M^{-1}y$ and $x\neq 0\iff y\neq 0\iff w_h(y)\geq 1.$\\
Then $h(M,x)=h(M^{-1},y)$ and

\begin{equation*}
    \begin{split}
        \mathcal B(M)=&\min\left\{\min\left\{h(M,x)\mid x\in \mathbb{F}_q^n, 1\leq w_h(x)\leq \left\lfloor \frac{n+1}{2}\right\rfloor\right\},\right.\\
        &\left.\min\left\{h(M^{-1},y)\mid x\in \mathbb{F}_q^n, 1\leq w_h(x)\leq n,  1\leq w_h(y)\leq \left\lfloor \frac{n+1}{2}\right\rfloor\right\}\right\}. 
    \end{split}
\end{equation*}

We may drop the condition $1\leq w_h(x)\leq n$ as this is a trivial condition for $x\neq 0.$ Note that the term $x\in \mathbb{F}_q^n$ may be replaced by $ y\in \mathbb{F}_q^n$ as the correspondence $x \to y$ is one-to-one. Therefore, we may write 

\begin{equation*}
    \begin{split}
        \mathcal B(M)=&\min\left\{\min\left\{h(M,x)\mid x\in \mathbb{F}_q^n, 1\leq w_h(x)\leq \left\lfloor \frac{n+1}{2}\right\rfloor\right\},\right.\\
        &\left.\min\left\{h(M^{-1},y)\mid y\in \mathbb{F}_q^n, 1\leq w_h(y)\leq \left\lfloor \frac{n+1}{2}\right\rfloor\right\}\right\}. 
    \end{split}
\end{equation*}

We may also change $y$ by $x$, and then we have 
\begin{equation*}
    \begin{split}
        \mathcal B(M)=&\min\left\{\min\left\{h(M,x)\mid x\in \mathbb{F}_q^n, 1\leq w_h(x)\leq \left\lfloor \frac{n+1}{2}\right\rfloor\right\},\right.\\
        &\left.\min\left\{h(M^{-1},x)\mid x\in \mathbb{F}_q^n, 1\leq w_h(x)\leq \left\lfloor \frac{n+1}{2}\right\rfloor\right\}\right\} 
    \end{split}
\end{equation*}
Or, 
\begin{equation}\label{Thm_branch_number_eqn:6}
    \begin{split}
        \mathcal B(M)=\min\left\{\min\left\{h(M,x),h(M^{-1},x)\right\}\mid x\in \mathbb{F}_q^n, 1\leq w_h(x)\leq \left\lfloor \frac{n+1}{2}\right\rfloor \right\}. 
    \end{split}
\end{equation}
This completes the proof. \qed
\end{proof} 

\noindent We now illustrate that it is unnecessary to consider all input vectors with weights up to $\floor{(n+1)/2}$. Instead, only a few select vectors are required. Specifically, we demonstrate that we can compute the branch number by concentrating on equivalence classes. Consequently, we can determine $\mathcal{B}(M)$ by calculating over the representatives of these equivalence classes, thereby enhancing the efficiency of our approach.

Define an equivalence relation $\sim$ in $\mathbb{F}_q^n$ as $a$ and $b$ in $\mathbb{F}_q^n$ are related if there exists $c\in \mathbb{F}_q^*$ such that $a=cb$, where $cb$ is the multiplication of the vector $b$ with a scalar $c$. For $a\in \mathbb{F}_q^n$, let $\bar{a}$ denote the  equivalence class of $a$.

Let $U_k= \{x \mid x \in \mathbb{F}_q^n, w_h(x)=k\}$ represent the set of all vectors of weight $k$ in $\mathbb{F}_q^n$. We now decompose $U_k$ into equivalence classes. The set of representative elements for each equivalence class of $U_k$ is denoted by $S_k$. The $S_k$ is defined as follows:
$$S_k =
\{x \in U_k \mid x=(a_1,a_2 \ldots a_n), \text{the first non-zero coordinate in }x \text{ is } 1\}.$$

\begin{lemma} \label{lem_representative_elements}
The following requirements are met by $U_k$ and $S_k$ for $k=1,2,\ldots n$:
\begin{enumerate}
\item[(i)] Given $x,y\in S_k, \bar{x}=\bar{y}\implies x=y$. In other words, any two distinct elements of $S_k$ belong to distinct equivalence classes. 
\item[(ii)] $U_k=\underset{x\in S_k}{{\bigcup}}\bar{x}$, where $\bar{x}$ denotes the equivalence class of $x$, i.e., every equivalence class contains an element of $S_k$.
\end{enumerate}
\end{lemma}
\begin{proof}
Let $x,y \in S_k$ be such that $x\neq y$ and suppose $x$ and $y$ belong to same equivalence class, i.e, $\exists$ $c\in \mathbb{F}_q^*$ such that $x=cy$. 
Let $x=(x_1,x_2,\ldots,x_n)$ and $y=(y_1,y_2,\ldots,y_n)$ and let $t$ be the least index such that $x_t\neq 0$ and $y_t\neq 0$. Then $x_t=1$ and $y_t=1$. Consequently, $x$ and $y$ can be written as
\begin{equation*}
x=(0,0,\ldots0,1,x_{t+1},\ldots,x_n),~y=(0,0,\ldots0,1,y_{t+1},\ldots,y_n).
\end{equation*}

As $x=cy$, we obtain  $(0,0,\ldots0,1,x_{t+1},\ldots,x_n)=(0,0,\ldots0,c,cy_{t+1},\ldots,cy_n)$. After comparing these, we get $c=1$, which further yields $x=y$.   This proves part $(i)$.

\noindent Next, we prove part $(ii)$, i.e., $U_k=\underset{x\in S_k}{{\bigcup}}\bar{x}$.

\begin{eqnarray*}
&& \text{Let } \alpha  \in \underset{x\in S_k}{{\bigcup}}\bar{x}\\
&\Rightarrow & \exists \;y \in S_k  \text{ such that }  \alpha  \in \bar{y}.\end{eqnarray*} This further implies that   \begin{eqnarray*}
&   & \alpha= cy, \text{where }  c  \in \mathbb{F}_q^* \\
& \Rightarrow & w_h(\alpha)= w_h(y)=k  \\
& \Rightarrow & \alpha \in  U_k  \\
& \Rightarrow &  \underset{x\in S_k}{{\bigcup}}\bar{x} \subseteq U_k.
\end{eqnarray*}

Now, let $\beta \in U_k$. Then $\beta \in  \mathbb{F}_q^n$ is such that $w_h(\beta)=k$.\\
Let $\beta =(\beta_1,\beta_2,\ldots,\beta_n)$ and $i$ be the least index for which $\beta _i$ is non-zero. 
Let $\beta_1=c\beta$, where $c=\beta_i^{-1}$. Then $\beta\sim \beta_1$ and $\beta_1=(0,0,\ldots,1,\beta_{i+1},\beta_{i+2},\ldots,\beta_n)$.
\begin{eqnarray*}
&& \text{ Clearly } \beta_1 \in S_k \text{ and } \beta\in\bar{\beta_1}.\\
&\Rightarrow & \beta \in \underset{x\in S_k}{{\bigcup}}\bar{x}\\
&\Rightarrow & U_k \subseteq \underset{x\in S_k}{{\bigcup}}\bar{x} .
\end{eqnarray*}
Therefore, we have that $U_k = \underset{x\in S_k}{{\bigcup}}\bar{x}$. This completes the proof. \qed
\end{proof}

\begin{theorem} \label{thm_branch_no}
The branch number of an invertible matrix $M\in M_n(\mathbb{F}_q)$ is given as
\[
\mathcal B(M)=\min\left\{\min\left\{h(M,x),h(M^{-1},x)\right\}\mid x\in  S_k, k=1,2,\ldots, \left\lfloor \frac{n+1}{2}\right\rfloor \right\}. 
\]
\end{theorem}
\begin{proof}
Using Lemma \ref{lem_representative_elements}, Equation (\ref{Thm_branch_number_eqn:6}) can be written as  
\begin{equation*}
\begin{split}
\mathcal B(M)=\min\left\{\min\left\{h(M,x),h(M^{-1},x)\right\}\mid x\in  U_k, k=1,2,\ldots, \left\lfloor \frac{n+1}{2}\right\rfloor \right\} \\ 
=\min\left\{\min\left\{h(M,x),h(M^{-1},x)\right\}\mid x \in  \underset{y\in S_k}{{\bigcup}}\bar{y}, k=1,2,\ldots,\left\lfloor \frac{n+1}{2}\right\rfloor\right\}. \end{split}
\end{equation*}\\
Given $1\leq k\leq \left\lfloor \frac{n+1}{2}\right\rfloor$ and $y\in S_k$, if $z\in \bar{y}$, then $\exists ~c\in \mathbb{F}_q^*$ such that $ z=cy$. Thus, we have 
\begin{equation*}
h(M,z)=h(M,cy)=h(M,y).   
\end{equation*}

Similarly,  $h(M^{-1},z)=h(M^{-1},y)$ for all $z\in \bar{y}$.\\

Hence, we deduce that
\begin{equation} \label{thm_branch_no_eqn}
\mathcal B(M)=\min\left\{\min\left\{h(M,x),h(M^{-1},x)\right\}\mid x\in  S_k, k=1,2,\ldots,\left\lfloor \frac{n+1}{2}\right\rfloor\right\}. 
\end{equation}

This completes the proof.\qed
\end{proof}

\begin{remark}
    It is worth mentioning that in the definition of $S_k$, the first non-zero coordinate of the elements is set to 1. However, this coordinate can be fixed to any non-zero value in $\mathbb{F}_q$. Setting the first non-zero coordinate of all representatives to 1 saves one multiplication during the algorithm's execution, as multiplying by 1 incurs no multiplication cost.
\end{remark}
 
\subsection{Description of the Proposed Algorithm}
We describe our algorithm as Algorithm \ref{Algorithm1}, which is essentially the algorithmic formulation of Theorem~\ref{thm_branch_no} with some additional filtering steps. Our proposed algorithm takes an invertible matrix and its size as input and computes its branch number. The notations used in the algorithm are consistent with those used in the theoretical treatment of the topic. 

\begin{algorithm} 
\caption[1]{Computation of branch number of a matrix over a finite field} \label{Algorithm1}
\begin{algorithmic}[1]
\Function{getbranchnumber}{$M,n$}
\State $\mathcal B(M)\gets n+1$ \Comment{Stores branch number of $M$}
\State $r\gets \lfloor n+1/2\rfloor$
\For{$(k\gets 1 \text{ to } r)$}
\While{$(S_k\neq \emptyset)$}
\State Choose $z\in S_k$ 
\State $w\gets$ weight of $(Mz)$ and $w'\gets$ weight of $(M^{-1}z)$ \label{mat_mul}
\If{$(w>w')$}
\State $w\gets w'$
\EndIf
\If{$(w+k<\mathcal B(M))$}
\State $\mathcal B(M)\gets w+k$
\EndIf
\If{$(\mathcal B(M)\leq r)$} \label{reduce_classes_1}
\State $r\gets \mathcal B(M)-1$ \label{reduce_classes_2}
\EndIf
\State $S_k\gets S_k\setminus \{z\}$
\EndWhile
\EndFor
\State\text{\bf return }$\mathcal B(M)$
\EndFunction
\end{algorithmic}
\end{algorithm}

\noindent In Example \ref{branch_no_exm}, we demonstrate how to calculate the branch number of an $8\times 8$ non-singular matrix over $\mathbb{F}_{2^8}$ using our proposed algorithm.

\begin{exm} \label{branch_no_exm}
Let $\mathbb{F}_{2^8}$ be the finite field generated by the primitive polynomial $x^8+x^4+x^3+x^2+1$. Now, consider the $8\times 8$ involutory MDS matrix $M$ over $\mathbb{F}_{2^8}$ provided in \cite{barreto2000khazad}.

\begin{center}
$M=\begin{pmatrix}
01_x & 03_x & 04_x & 05_x & 06_x & 08_x & 0B_x & 07_x \\
03_x & 01_x & 05_x & 04_x & 08_x & 06_x & 07_x & 0B_x \\
04_x & 05_x & 01_x & 03_x & 0B_x & 07_x & 06_x & 08_x \\
05_x & 04_x & 03_x & 01_x & 07_x & 0B_x & 08_x & 06_x \\
06_x & 08_x & 0B_x & 07_x & 01_x & 03_x & 04_x & 05_x \\
08_x & 06_x & 07_x & 0B_x & 03_x & 01_x & 05_x & 04_x \\
0B_x & 07_x & 06_x & 08_x & 04_x & 05_x & 01_x & 03_x \\
07_x & 0B_x & 08_x & 06_x & 05_x & 04_x & 03_x & 01_x	   
\end{pmatrix}$.
\end{center}
We find the branch number of matrix $M$ using our algorithm in the following steps.\\
First, we find the inverse of $M$ i.e. $M^{-1}=M$. Here, $r=\lfloor n+1/2\rfloor=\lfloor 9/2\rfloor=4$. For computing branch number of $M$, our search space is $\cup_{k=1}^{4} S_k$. The cardinality of $S_k$ represents the number of vectors in $\mathbb{F}_{2^8}^8$ of weight $k$ with the first non-zero coordinate equal to $1$, i.e. $|S_k|= \binom{8}{k}(2^8-1)^{k-1}$. Thus, we have \\

$\mathbf{|S_1|}=8$, $\mathbf{|S_2|}=28\times (2^8-1)$, $\mathbf{|S_3|}=56\times (2^8-1)^2$, and $\mathbf{|S_4|}=70\times (2^8-1)^3$.
\\
\noindent We need to search for $\min \set{ \min \set{w_h(x)+w_h(Mx)}, \min \set{w_h(x)+w_h(M^{-1}x)}}$ for the computation of branch number. Since $M^{-1}=M$, therefore, we search only for the minimum value of $w_h(x)+w_h(Mx)$ in $S_1$,~$S_2$,~$S_3$, and $S_4$.\\

\noindent In $\mathbf{S_1}$: ~$\min\left\{w_h(x)+w_h(Mx)\right\}=9$ \\
In $\mathbf{S_2}$: ~$\min\left\{w_h(x)+w_h(Mx)\right\}=9$ \\
In $\mathbf{S_3}$: ~$\min\left\{w_h(x)+w_h(Mx)\right\}=9$ \\
In $\mathbf{S_4}$: ~$\min\left\{w_h(x)+w_h(Mx)\right\}=9$. \\

Hence, the branch number of $M$ is $9$.
\end{exm}

\begin{remark}\label{Remark_Filtering_classes}
Steps \ref{reduce_classes_1} and \ref{reduce_classes_2} of Algorithm \ref{Algorithm1} are designed to filter some classes. If at any point the value of $\mathcal{B}(M)$ turns out to be less than or equal to $r$, the algorithm skips all classes with weights greater than or equal to $\mathcal{B}(M)$. This is because when $k$ equals $\mathcal{B}(M)$, the value of $w+k$ is greater than or equal to $\mathcal{B}(M)+1$, which is greater than $\mathcal{B}(M)$. While these steps do not alter the outcome of the algorithm, they contribute to reducing its average-case complexity.
\end{remark}

\noindent In the following example, we discuss the advantages of Steps \ref{reduce_classes_1} and \ref{reduce_classes_2} in Algorithm \ref{Algorithm1}.

\begin{exm} 
Let $\mathbb{F}_{2^8}$ be the finite field generated by the primitive polynomial $x^8+x^4+x^3+x^2+1$. Now, consider the $8\times 8$ matrix $M$ over $\mathbb{F}_{2^8}$.

\begin{center}
$M=\begin{pmatrix}
01_x & 02_x & 03_x & 04_x & 01_x & 02_x & 03_x & 07_x \\
02_x & 01_x & 04_x & 03_x & 02_x & 01_x & 07_x & 03_x \\
03_x & 04_x & 01_x & 02_x & 03_x & 07_x & 01_x & 02_x \\
04_x & 03_x & 02_x & 01_x & 07_x & 03_x & 02_x & 01_x \\
01_x & 02_x & 03_x & 07_x & 01_x & 02_x & 03_x & 04_x \\
02_x & 01_x & 07_x & 03_x & 02_x & 01_x & 04_x & 03_x \\
03_x & 07_x & 01_x & 02_x & 03_x & 04_x & 01_x & 02_x \\
07_x & 03_x & 02_x & 01_x & 07_x & 03_x & 02_x & 01_x	   
\end{pmatrix}$.
\end{center}

Here, $r= \lfloor n+1/2\rfloor=\lfloor 9/2\rfloor=4$, and initially, our search space is $\cup_{k=1}^{4} S_k$. The cardinality of $S_k$ represents the number of vectors in $\mathbb{F}_{2^8}^8$ of weight $k$ with the first non-zero coordinate equal to $1$, i.e. $|S_k|= \binom{8}{k}(2^8-1)^{k-1}$.

First, we find the inverse of $M$ i.e. 
\begin{center}
$M^{-1}=\begin{pmatrix}
00_x & 00_x & 00_x & f4_x & 00_x & 00_x & 00_x & f4_x \\
ce_x & 80_x & 1c_x & 00_x & ce_x & 80_x & e8_x & ba_x \\
e9_x & 9c_x & 80_x & 00_x & e9_x & 68_x & 80_x & 27_x \\
26_x & e9_x & ce_x & 00_x & d2_x & e9_x & ce_x & 9d_x \\
9d_x & 27_x & ba_x & f4_x&9d_x & 27_x & ba_x & 4e_x \\
ce_x & 80_x & e8_x & 00_x & ce_x & 80_x & 1c_x & ba_x \\
e9_x & 68_x & 80_x & 00_x & e9_x & 9c_x & 80_x & 27_x \\
d2_x & e9_x & ce_x & 00_x & 26_x & e9_x & ce_x & 9d_x	   
\end{pmatrix}$.
\end{center}

\noindent We need to search for $\min \set{ \min \set{w_h(x)+w_h(Mx)}, \min \set{w_h(x)+w_h(M^{-1}x)}}$. In $S_1$, we find $\min \set{ \min \set{w_h(x)+w_h(Mx)}, \min \set{w_h(x)+w_h(M^{-1}x)}}=3$, which is less than the initial value of $r=4$. According to Remark~\ref{Remark_Filtering_classes}, the updated value of $r$ becomes $3-1=2$, and thus, we need to search over $\cup_{k=1}^{2} S_k$.
\\
\noindent Now, in $S_2$, we have $\min \set{ \min \set{w_h(x)+w_h(Mx)}, \min \set{w_h(x)+w_h(M^{-1}x)}}=3$. Therefore, we conclude that the branch number of $M$ is $3$.
\end{exm}

\section{Complexity Analysis} \label{sec:computational_complexity}
In this section, we analyze the complexity of the proposed algorithm and compare it with that of the exhaustive approach. The algorithm is not memory-intensive; it requires both $M$ and its inverse to be stored in memory. For a matrix of order $n$, a maximum of $2n^2$ values should be stored. Although this is twice as large as what is required for the exhaustive approach, it is typically not a concern because we usually deal with matrices of small orders, such as 4, 8, 16, or 32. In other words, space complexity is not a significant issue here. Consequently, our analysis focuses on time complexity. 

The proposed Algorithm \ref{Algorithm1} is an algorithmic formulation of Equation (\ref{thm_branch_no_eqn}). The computationally dominant step in the algorithm is Step \ref{mat_mul}, which is repeated $\sum_{k=1}^{r}o(S_k)$ times, where $r=\lfloor n+1/2\rfloor$. Given $1\leq k\leq r$, step \ref{mat_mul} involves two matrix multiplications. This will take $2n(k-1)$ multiplications (we have $(k-1)$ in place of $k$ as the first non-zero coordinate in every vector of $S_k$ is always 1). 
Furthermore, note that $o(S_k)=\binom{n}{k}(q-1)^{k-1}$ as $S_k$ contains vectors of weight $k$ from $\mathbb{F}_q^n$ whose first non-zero coordinate is $1$.
Thus, the computational complexity is of the order of 
\begin{equation}\label{complexity_eqn}
2n\sum_{k=1}^{r}o(S_k)\cdot(k-1)=2n\sum_{k=1}^{\left\lfloor\frac{n+1}{2}\right\rfloor}\binom{n}{k}(k-1)(q-1)^{k-1}.
\end{equation}

Next, we demonstrate that this complexity is much lower than that required for the exhaustive approach. Before proceeding to the actual comparison, we establish some necessary results.

\begin{prp}\label{prp_complexity_naive}
The computational complexity of computing the branch number of an $n\times n$ matrix over $\mathbb{F}_q$ using the exhaustive approach is $O(n^2q^n)$.
\end{prp}
\begin{proof}
The exhaustive approach is based on the definition of branch number, i.e.,
\[\mathcal B(M)=\min\{w_h(x)+w_h(Mx)\mid x\in \mathbb{F}_q^n, x \neq 0\}.\]
The dominating step involved here is the matrix multiplication step. One matrix multiplication roughly takes $n^2$ field additions and $n^2$ field multiplications. The complexity of matrix multiplication is $O(n^2)$. To compute the minimum, total $q^n-1$ matrix multiplications are required. Thus, the complexity becomes $O(n^2q^n)$. \qed
\end{proof}

\begin{remark}
To the best of our knowledge, there is no algorithm in the literature for computing the branch number of any non-singular matrix over $\mathbb{F}_q$, except for the exhaustive approach. Therefore, our goal is to compare the complexity with that of an exhaustive approach. However, the complexity of the exhaustive approach appears quite different from the expression of the complexity of our approach, as given in (\ref{complexity_eqn}). So, it is not straightforward to compare the two complexities. We simplify the expression (\ref{complexity_eqn}) in Theorem \ref{complexity_thm}, and finally, we compare the simplified expression with the complexity of the exhaustive approach.
\end{remark}

\noindent In Theorem \ref{complexity_thm}, we simplify the expression (\ref{complexity_eqn}) for comparison with our algorithm. To do this, we require the following lemma:

\begin{lemma} \label{Lemma_combinatorial}
For $n\geq 1$ and $q>2$, we have
\begin{equation*}
\sum_{k=1}^{\lfloor\frac{n-1}{2}\rfloor}\binom{n}{k}(k-1)(q-1)^{k-1}\leq \binom{n}{\lfloor\frac{n+1}{2}\rfloor}\left \lfloor\frac{n-1}{2}\right \rfloor  (q-1)^{\lfloor\frac{n-1}{2}\rfloor}.
\end{equation*}
\end{lemma}

\begin{proof}
We prove the lemma by induction on $n$. Let $P(n)$ denote the statement for any integer $n\geq 1$. For $n=1$, both sides are zero. Thus, $P(1)$ is true. Now we assume that all $P(1),\ldots, P(n)$ are true. We aim to show that $P(n+1)$ is also true. 

We have
\begin{eqnarray*}	
&&\sum_{k=1}^{\lfloor\frac{n}{2}\rfloor}\binom{n+1}{k}(k-1)(q-1)^{k-1}\\
&=&\sum_{k=1}^{\lfloor\frac{n-2}{2}\rfloor}\binom{n+1}{k}(k-1)(q-1)^{k-1}+\binom{n+1}{\lfloor\frac{n}{2}\rfloor}\left \lfloor\frac{n-2}{2}\right\rfloor(q-1)^{\lfloor\frac{n-2}{2}\rfloor}\\
&=&\sum_{k=1}^{\lfloor\frac{n-2}{2}\rfloor}\frac{(n+1)n}{(n+1-k)(n-k)}\binom{n-1}{k}(k-1)(q-1)^{k-1}\\
&&+\binom{n+1}{\lfloor\frac{n}{2}\rfloor}\left\lfloor\frac{n-2}{2}\right\rfloor(q-1)^{\lfloor\frac{n-2}{2}\rfloor}.
\end{eqnarray*}

Note that $(n+1-k)(n-k)\geq (n+1-\lfloor\frac{n-2}{2}\rfloor)(n-\lfloor\frac{n-2}{2}\rfloor)$ for $k=1,2,\ldots \lfloor\frac{n-2}{2}\rfloor$. Thus, from above we have

\begin{eqnarray*}
&&\sum_{k=1}^{\lfloor\frac{n}{2}\rfloor}\binom{n+1}{k}(k-1)(q-1)^{k-1} \\
&\leq&\frac{(n+1)n}{(n+1-\lfloor\frac{n-2}{2}\rfloor)(n-\lfloor\frac{n-2}{2}\rfloor)}\sum_{k=1}^{\lfloor\frac{n-2}{2}\rfloor}\binom{n-1}{k}(k-1)(q-1)^{k-1} \\
&&+\binom{n+1}{\lfloor\frac{n}{2}\rfloor}\left\lfloor\frac{n-2}{2}\right\rfloor(q-1)^{\lfloor\frac{n-2}{2}\rfloor}.
\end{eqnarray*}

Since  $(n+1-\lfloor\frac{n-2}{2}\rfloor)(n-\lfloor\frac{n-2}{2}\rfloor) \geq (n+1-\lfloor\frac{n}{2}\rfloor)(n-\lfloor\frac{n}{2}\rfloor)$. Consequently, from above we may write 
\begin{eqnarray}
&&\sum_{k=1}^{\lfloor\frac{n}{2}\rfloor}\binom{n+1}{k}(k-1)(q-1)^{k-1} \notag \\
&\leq&\frac{(n+1)n}{(n+1-\lfloor\frac{n}{2}\rfloor)(n-\lfloor\frac{n}{2}\rfloor)}\sum_{k=1}^{\lfloor\frac{n-2}{2}\rfloor}\binom{n-1}{k}(k-1)(q-1)^{k-1} \notag \\
&&+\binom{n+1}{\lfloor\frac{n}{2}\rfloor}\left\lfloor\frac{n-2}{2}\right\rfloor(q-1)^{\lfloor\frac{n-2}{2}\rfloor} \label{Lemma_combinatorial_eqn_1}. 
\end{eqnarray}
Since by induction hypothesis $P(n-1)$ is true, we have
\begin{equation}\label{Lemma_combinatorial_eqn_2}
\sum_{k=1}^{\lfloor\frac{n-2}{2}\rfloor}\binom{n-1}{k}(k-1)(q-1)^{k-1}\leq \binom{n-1}{\lfloor\frac{n}{2}\rfloor}\left \lfloor\frac{n-2}{2}\right \rfloor  (q-1)^{\lfloor\frac{n-2}{2}\rfloor}.
\end{equation}

Therefore, from (\ref{Lemma_combinatorial_eqn_1}) and (\ref{Lemma_combinatorial_eqn_2}), we have  
\begin{eqnarray*}
&&\sum_{k=1}^{\lfloor\frac{n}{2}\rfloor}\binom{n+1}{k}(k-1)(q-1)^{k-1} \\
&\leq&\frac{(n+1)n}{(n+1-\lfloor\frac{n}{2}\rfloor)(n-\lfloor\frac{n}{2}\rfloor)} \binom{n-1}{\lfloor\frac{n}{2}\rfloor}\left\lfloor\frac{n-2}{2}\right\rfloor(q-1)^{\lfloor\frac{n-2}{2}\rfloor} \\
&&+\binom{n+1}{\lfloor\frac{n}{2}\rfloor}\left\lfloor\frac{n-2}{2}\right\rfloor(q-1)^{\lfloor\frac{n-2}{2}\rfloor} \\
&\leq&\binom{n+1}{\lfloor\frac{n}{2}\rfloor}\left\lfloor\frac{n-2}{2}\right\rfloor(q-1)^{\lfloor\frac{n-2}{2}\rfloor}+\binom{n+1}{\lfloor\frac{n}{2}\rfloor}\left\lfloor\frac{n-2}{2}\right\rfloor(q-1)^{\lfloor\frac{n-2}{2}\rfloor} \\
&\leq&2\binom{n+1}{\lfloor\frac{n}{2}\rfloor}\left\lfloor\frac{n-2}{2}\right\rfloor(q-1)^{\lfloor\frac{n-2}{2}\rfloor}. 
\end{eqnarray*}	
Note that $q>2 \implies q-1\geq 2$, and $\binom{n}{r}=\binom{n}{n-r}$, for any positive integer $r\leq n$. Thus, we have 
\begin{eqnarray}
&&\sum_{k=1}^{\lfloor\frac{n}{2}\rfloor}\binom{n+1}{k}(k-1)(q-1)^{k-1} \leq\binom{n+1}{\lfloor\frac{n+2}{2}\rfloor}\left\lfloor\frac{n}{2}\right\rfloor(q-1)^{\lfloor\frac{n}{2}\rfloor}.\nonumber
\end{eqnarray}
This shows that $P(n+1)$ is also true. Thus, the result holds for all integers $n\geq 1$ and $q>2$. \qed
\end{proof}

\begin{theorem} \label{complexity_thm}
The computational complexity of Algorithm~\ref{Algorithm1} is 
$$O(2^{n+\frac{3}{2}+\frac{n-1}{2}\log_2(q-1)+\frac{3}{2}\log_2 n}).$$
\end{theorem}
\begin{proof}
From (\ref{complexity_eqn}), we know the computational complexity of Algorithm~\ref{Algorithm1} is

\begin{equation*}
    \begin{aligned}
        & 2n\sum_{k=1}^{\lfloor\frac{n+1}{2}\rfloor}\binom{n}{k}(k-1)(q-1)^{k-1}.
    \end{aligned}
\end{equation*}

Now, we can write the above expression as:
\begin{eqnarray*}
&&2n\sum_{k=1}^{\lfloor\frac{n+1}{2}\rfloor}\binom{n}{k}(k-1)(q-1)^{k-1}\\
&=& 2n\sum_{k=1}^{\lfloor\frac{n-1}{2}\rfloor}\binom{n}{k}(k-1)(q-1)^{k-1} + 2n\binom{n}{\lfloor\frac{n+1}{2}\rfloor}\left\lfloor\frac{n-1}{2}\right\rfloor (q-1)^{\lfloor\frac{n-1}{2}\rfloor}.
\end{eqnarray*}

From Lemma~\ref{Lemma_combinatorial}, we have
\begin{equation}\label{complexity_thm_eqn_1}
\begin{aligned}
    2n\sum_{k=1}^{\lfloor\frac{n+1}{2}\rfloor}\binom{n}{k}(k-1)(q-1)^{k-1} 
    & \leq 4n\binom{n}{\lfloor\frac{n+1}{2}\rfloor}\left\lfloor\frac{n-1}{2}\right\rfloor (q-1)^{\lfloor\frac{n-1}{2}\rfloor}\\
    & = 4n\binom{n}{\lfloor\frac{n}{2}\rfloor}\left\lfloor\frac{n-1}{2}\right\rfloor (q-1)^{\lfloor\frac{n-1}{2}\rfloor}
\end{aligned}
\end{equation}

Now, we will employ the following inequality~\footnote{By the Stirling's approximation, we have $r! \approx \sqrt{2\pi r} \left( \frac{r}{e} \right)^r$. Thus, $${2r \choose r} \approx \frac{\sqrt{4\pi r} \left( \frac{2r}{e} \right)^{2r}}{\sqrt{2\pi r} \left( \frac{r}{e} \right)^r \cdot \sqrt{2\pi r} \left( \frac{r}{e} \right)^r}= \frac{2^{2r}}{\sqrt{\pi r}}< \frac{2^{2r}}{\sqrt{r}}.$$}:

\begin{equation}\label{complexity_thm_eqn_2}
{2r \choose r} < \frac{2^{2r}}{\sqrt{r}},
\end{equation}
where $r$ is any positive integer.

Therefore, from (\ref{complexity_thm_eqn_1}) and (\ref{complexity_thm_eqn_2}), we have

\begin{eqnarray*}
2n\sum_{k=1}^{\lfloor\frac{n+1}{2}\rfloor}\binom{n}{k}(k-1)(q-1)^{k-1}&<& 4n\frac{2^{n}}{\sqrt{\lfloor\frac{n}{2}\rfloor}}\left\lfloor\frac{n-1}{2}\right\rfloor (q-1)^{\lfloor\frac{n-1}{2}\rfloor}\nonumber\\
&<& 2^{n+3/2}n^{3/2} (q-1)^{\frac{n-1}{2}}\nonumber\\
&=& 2^{n+\frac{3}{2}+\frac{n-1}{2}\log_2(q-1)+\frac{3}{2}\log_2 n}.
\end{eqnarray*}
Therefore, we estimate that
\begin{equation*}
2n\sum_{k=1}^{\lfloor\frac{n+1}{2}\rfloor}\binom{n}{k}(k-1)(q-1)^{k-1}=O(2^{n+\frac{3}{2}+\frac{n-1}{2}\log_2(q-1)+\frac{3}{2}\log_2 n}).
\end{equation*}
This completes the proof. \qed
\end{proof}

We will now theoretically demonstrate that the complexity of Algorithm~\ref{Algorithm1} is significantly lower than that of the naive approach when $n,q\geq 4$.

\begin{theorem} \label{thm:4.3}
For $n,q\geq 4$, the complexity of Algorithm~\ref{Algorithm1} is significantly lower than that of the naive approach.
\end{theorem}

\begin{proof}
Proposition~\ref{prp_complexity_naive} states that the complexity of the naive approach is $O(n^2q^n)$, which is equivalent to $O(2^{n\log_2 q +2\log_2 n})$. On the other hand, according to Theorem~\ref{complexity_thm}, the complexity of Algorithm~\ref{Algorithm1} is $O(2^{n+\frac{3}{2}+\frac{n-1}{2}\log_2(q-1)+\frac{3}{2}\log_2 n})$.\\

Now let us define 
\begin{eqnarray*}
f(n,q)&=&(n\log_2 q +2\log_2 n)-(n+\frac{3}{2}+\frac{n-1}{2}\log_2(q-1)+\frac{3}{2}\log_2 n)\\
&=&\frac{n-1}{2}\log_2 q+\frac{n+1}{2}\log_2 q-(n+\frac{3}{2}+\frac{n-1}{2}\log_2(q-1))+\frac{1}{2}\log_2 n\\
&=&\frac{n-1}{2}\log_2 \frac{q}{q-1}+\frac{n+1}{2}\log_2 q-\left(n+\frac{3}{2}\right)+\frac{1}{2}\log_2 n.
\end{eqnarray*}

We observe that this expression is increasing in $n$ for all $n>0$. We show that it also increases in $q$ when $q>2$.

We differentiate the function $f(n,q)$ partially with respect to $q$ to get
\begin{eqnarray*}
\frac{\partial}{\partial q}f(n,q)&=&-\frac{n-1}{2}\left(\frac{q-1}{q}.\frac{1}{(q-1)^2}\right)+\frac{n+1}{2q}\\
&=&\frac{n-1}{2q}\left(1-\frac{1}{q-1}\right)+\frac{1}{q}.
\end{eqnarray*}
This is positive for $n\geq 1$ and $q>2$ and hence $f(n,q)$ is increasing. Furthermore, We note that\\
$$f(4,4)=\frac{3}{2}\log_2 \frac{4}{3}+\frac{5}{2}\log_2 4-\left(4+\frac{3}{2}\right)+\frac{1}{2}\log_2 4$$ $$=\frac{3}{2}\log_2 \frac{4}{3}+\frac{1}{2}>0. \hspace{22mm}$$
Hence $f(n,q)> 0$ for all $n\geq 4$ and $q\geq 4$.
Thus, the theorem is proved. \qed
\end{proof}

%\begin{remark}
%We note that the complexity of the direct approach is $O(n^2q^n)$, which can also be written as $O(2^{n\log_2 q +2\log_2 n})$. It is clear from Theorem \ref{thm:4.3} that the complexity of Algorithm~\ref{Algorithm1} is less than that of the naive approach when $n,q\geq 4$.
%\end{remark}

\subsection{A Special Type of Non-Singular Matrices}
For an involutory matrix $M$, we have $M^{-1}=M$. Also, for the Hadamard matrix~\footnote{A $2^n \times 2^n$ matrix $M$ is Hadamard matrix in $\mathbb{F}_{2^r}$ if it can be expressed in the form:
$\begin{bmatrix}
U & V \\
V & U
\end{bmatrix}$,
where the submatrices $U$ and $V$ are also Hadamard matrices.} over characteristic $2$, we have $M^{-1}= \frac{1}{c^2} M$, where $c$ is the sum of the elements in the first row. Therefore, for a non-singular Hadamard matrix or any involutory matrix, we have $w_h(Mx)= w_h(M^{-1}x)$. Thus, the complexity of our proposed algorithm will be further reduced. Now, step \ref{mat_mul} of Algorithm \ref{Algorithm1} requires only one matrix multiplication, which will take $n(k-1)$ multiplications. Thus, for a non-singular Hadamard or any involutory matrix, the computational complexity is of the order of 

\begin{equation*}
n\sum_{k=1}^{r}o(S_k)\cdot(k-1)=n\sum_{k=1}^{\left\lfloor\frac{n+1}{2}\right\rfloor}\binom{n}{k}(k-1)(q-1)^{k-1}.
\end{equation*}

Now, similar to the proof of Theorem \ref{complexity_thm}, we can use the above expression to determine the computational complexity for a non-singular Hadamard matrix or any involutory matrix as follows:

\begin{theorem}
The computational complexity of Algorithm~\ref{Algorithm1} for a non-singular Hadamard or any involutory matrix is 
$$O(2^{n+\frac{1}{2}+\frac{n-1}{2}\log_2(q-1)+\frac{3}{2}\log_2 n}).$$
\end{theorem}

\noindent To validate the effectiveness of our algorithm, we conducted extensive computational experiments on various non-singular matrices over different finite fields using (\ref{complexity_eqn}). The results demonstrate a notable improvement in computational efficiency compared to an exhaustive approach, especially for large matrices and finite fields. Furthermore, the advantage of the proposed algorithm over the exhaustive approach can be quantitatively seen in Table \ref{tab:exp_results} for different values of $n$ and $q$. The size $q$ of the finite field for every value of $n$ is divided into $2^8$ and $2^{16}$.\\

Table \ref{tab:exp_results} shows that the number of field multiplications required for the branch number computation is significantly reduced compared to the exhaustive approach. Using this algorithm, the computation for the branch number for certain values of $n$ and $q$ (e.g., for $n=6$ and $q=2^{16}$) can be done in real-time, which would otherwise not be computationally feasible through the exhaustive approach.

\begin{remark}
It is easy to verify that our algorithm and its mathematical foundations apply to both differential and linear branch numbers (see Definitions~\ref{Def_diff_branch} and ~\ref{Def_lin_branch}). This can be done by simply changing $M$ to $M^T$.
\end{remark}

\begin{table}[h] 
\centering
\renewcommand{\arraystretch}{1.25}
\caption{Comparison of the computational complexity of the proposed algorithm over the exhaustive approach}\label{tab:exp_results}
\begin{tabular}{|c|c|c|c|} 
\hline
&  & \multicolumn{2}{p{4.8cm}|}{\centering {\bf Number of field multiplications}}\\
\cline{3-4}
{\bf Matrix size $(n)$} &{\bf Field size $(q)$}	& \multicolumn{1}{p{2.4cm}|}{\centering {\bf Exhaustive Approach}}
& \multicolumn{1}{p{2.4cm}|}{\centering{\bf Proposed Algorithm}}\\
\hline
\multirow{2}{1.5cm}{\centering 4} 
&$2^8$ &$2^{36}$ & $2^{13.58}$ \\
\cline{2-4}
&$2^{16}$ &$2^{68}$ & $2^{21.58}$\\
\hline	
\multirow{2}{1.5cm}{\centering 5} &$2^{8}$ &$2^{44.64}$ & $2^{23.63}$\\
\cline{2-4}
&$2^{16}$ &$2^{84.64}$ & $2^{39.64}$\\
\hline
\multirow{2}{1.5cm}{\centering 6} &$2^{8}$ &$2^{53.17}$ & $2^{24.89}$\\
\cline{2-4}
&$2^{16}$ &$2^{101.17}$ & $2^{40.91}$\\
\hline
\multirow{2}{1.5cm}{\centering 7} &$2^{8}$ &$2^{61.61}$ & $2^{34.51}$\\
\cline{2-4}
&$2^{16}$ &$2^{117.62}$ & $2^{58.52}$\\
\hline
\multirow{2}{1.5cm}{\centering 8} &$2^{8}$ &$2^{70}$ & $2^{35.70}$\\
\cline{2-4}
&$2^{16}$ &$2^{134}$ & $2^{59.71}$\\
\hline
\end{tabular}
\end{table}

\section{Conclusion} \label{sec:conclusion}
This paper presents an algorithm designed to compute the branch number of a non-singular matrix over a finite field. It proceeds to analyze the complexity of the proposed algorithm and demonstrates a notable reduction in computational complexity compared to the naive approach, approximately to the square root of the original requirement. Through detailed complexity analysis and a comparative table, it illustrates the significant improvement offered by the proposed algorithm over the naive approach. Additionally, given the connection highlighted in Remark~\ref{Remark_min_dis_is_branch_number} between the minimum distance of a linear code and the branch number, this advancement has the potential to offer insight into error-correcting codes, as well as its impact on symmetric key cryptographic primitives.

\subsubsection*{Acknowledgments:} We would like to express our sincere gratitude to the anonymous reviewers for their detailed and insightful feedback, which significantly contributed to the improvement of this article.

\bibliographystyle{acm}
\bibliography{ref}

\begin{thebibliography}{10}

\bibitem{PRIDE}
{\sc Albrecht, M.~R., Driessen, B., Kavun, E.~B., Leander, G., Paar, C., and
  Yal{\c{c}}{\i}n, T.}
\newblock Block {C}iphers -- {F}ocus on the {L}inear {L}ayer (feat. {PRIDE}).
\newblock In {\em Advances in Cryptology -- CRYPTO 2014\/} (Berlin, Heidelberg,
  2014), J.~A. Garay and R.~Gennaro, Eds., Springer Berlin Heidelberg,
  pp.~57--76.

\bibitem{MIDORI}
{\sc Banik, S., Bogdanov, A., Isobe, T., Shibutani, K., Hiwatari, H., Akishita,
  T., and Regazzoni, F.}
\newblock Midori: {A} {B}lock {C}ipher for {L}ow {E}nergy.
\newblock In {\em Advances in Cryptology -- ASIACRYPT 2015\/} (Berlin,
  Heidelberg, 2015), T.~Iwata and J.~H. Cheon, Eds., Springer Berlin
  Heidelberg, pp.~411--436.

\bibitem{barreto2000khazad}
{\sc Barreto, P., and Rijmen, V.}
\newblock The khazad legacy-level block cipher.
\newblock {\em Primitive submitted to NESSIE 97}, 106 (2000).

\bibitem{SKINNY}
{\sc Beierle, C., Jean, J., K{\"o}lbl, S., Leander, G., Moradi, A., Peyrin, T.,
  Sasaki, Y., Sasdrich, P., and Sim, S.~M.}
\newblock The {SKINNY} {F}amily of {B}lock {C}iphers and {I}ts {L}ow-{L}atency
  {V}ariant {MANTIS}.
\newblock In {\em Advances in Cryptology -- CRYPTO 2016\/} (Berlin, Heidelberg,
  2016), M.~Robshaw and J.~Katz, Eds., Springer Berlin Heidelberg,
  pp.~123--153.

\bibitem{biham1991differential}
{\sc Biham, E., and Shamir, A.}
\newblock Differential cryptanalysis of des-like cryptosystems.
\newblock {\em Journal of CRYPTOLOGY 4\/} (1991), 3--72.

\bibitem{Fides}
{\sc Bilgin, B., Bogdanov, A., Kne{\v{z}}evi{\'{c}}, M., Mendel, F., and Wang,
  Q.}
\newblock Fides: {L}ightweight {A}uthenticated cipher with {S}ide-{C}hannel
  {R}esistance for {C}onstrained {H}ardware.
\newblock In {\em Cryptographic Hardware and Embedded Systems - CHES 2013\/}
  (Berlin, Heidelberg, 2013), G.~Bertoni and J.-S. Coron, Eds., Springer Berlin
  Heidelberg, pp.~142--158.

\bibitem{PRINCE}
{\sc Borghoff, J., Canteaut, A., G{\"u}neysu, T., Kavun, E.~B., Knezevic, M.,
  Knudsen, L.~R., Leander, G., Nikov, V., Paar, C., Rechberger, C., Rombouts,
  P., Thomsen, S.~S., and Yal{\c{c}}{\i}n, T.}
\newblock {PRINCE} -- {A} {L}ow-{L}atency {B}lock {C}ipher for {P}ervasive
  {C}omputing {A}pplications.
\newblock In {\em Advances in Cryptology -- ASIACRYPT 2012\/} (Berlin,
  Heidelberg, 2012), X.~Wang and K.~Sako, Eds., Springer Berlin Heidelberg,
  pp.~208--225.

\bibitem{JDA_Thesis_1995}
{\sc Daemen, J.}
\newblock {\em Cipher and hash function design, strategies based on linear and
  differential cryptanalysis, {PhD} {T}hesis}.
\newblock K.U.Leuven, 1995.
\newblock \url{http://jda.noekeon.org/}.

\bibitem{AES}
{\sc Daemen, J., and Rijmen, V.}
\newblock {\em The {D}esign of {R}ijndael: {AES} - {T}he {A}dvanced
  {E}ncryption {S}tandard}.
\newblock Information Security and Cryptography. Springer, 2002.

\bibitem{MAELSTROM}
{\sc Filho, D. L.~G., Barreto, P.~S., and Rijmen, V.}
\newblock The {M}aelstrom-0 hash function.
\newblock {\em In Proceedings of the 6th Brazilian Symposium on Information and
  Computer Systems Security\/} (2006), 17--29.

\bibitem{GROSTL}
{\sc Gauravaram, P., Knudsen, L., Matusiewicz, K., Mendel, F., Rechberger, C.,
  Schläffer, M., and Thomsen, S.}
\newblock Grøstl - a {SHA}-3 candidate.
\newblock {\em Submission to {NIST}, 2008, Available at
  \url{http://www.groestl.info/}\/} (09 2008).

\bibitem{PHOTON}
{\sc Guo, J., Peyrin, T., and Poschmann, A.}
\newblock The {PHOTON} {F}amily of {L}ightweight {H}ash {F}unctions.
\newblock In {\em Advances in Cryptology -- CRYPTO 2011\/} (Berlin, Heidelberg,
  2011), P.~Rogaway, Ed., Springer Berlin Heidelberg, pp.~222--239.

\bibitem{guo2016constructing}
{\sc Guo, Z., Wu, W., and Gao, S.}
\newblock Constructing lightweight optimal diffusion primitives with feistel
  structure.
\newblock In {\em Selected Areas in Cryptography--SAC 2015: 22nd International
  Conference, Sackville, NB, Canada, August 12--14, 2015, Revised Selected
  Papers 22\/} (2016), Springer, pp.~352--372.

\bibitem{kcz}
{\sc Gupta, K.~C., Pandey, S.~K., Ray, I.~G., and Samanta, S.}
\newblock Cryptographically significant {MDS} matrices over finite fields: {A}
  brief survey and some generalized results.
\newblock {\em Advances in Mathematics of Communications 13}, 4 (2019),
  779--843.

\bibitem{Gupta2023direct}
{\sc Gupta, K.~C., Pandey, S.~K., and Samanta, S.}
\newblock On the {D}irect {C}onstruction of {MDS} and {N}ear-{MDS} {M}atrices.
\newblock arXiv:~2306.12848, 2023.
\newblock \url{https://arxiv.org/abs/2306.12848}.

\bibitem{NMDS_Gupta2024}
{\sc Gupta, K.~C., Pandey, S.~K., and Samanta, S.}
\newblock On the construction of near-{MDS} matrices.
\newblock {\em Cryptography and Communications 16}, 2 (Mar 2024), 249--283.

\bibitem{Li_Wang_2017}
{\sc Li, C., and Wang, Q.}
\newblock Design of lightweight linear diffusion layers from near-mds matrices.
\newblock {\em IACR Transactions on Symmetric Cryptology 2017}, 1 (Mar. 2017),
  129--155.

\bibitem{matsui1993linear}
{\sc Matsui, M.}
\newblock Linear cryptanalysis method for des cipher.
\newblock In {\em Workshop on the Theory and Application of of Cryptographic
  Techniques\/} (1993), Springer, pp.~386--397.

\bibitem{shan}
{\sc Shannon, C.~E.}
\newblock Communication {T}heory of {S}ecrecy {S}ystems.
\newblock {\em The {B}ell {S}ystem {T}echnical Journal 28}, 4 (1949), 656--715.

\bibitem{NMDS_code_2022}
{\sc Sui, J., Yue, Q., Li, X., and Huang, D.}
\newblock {MDS}, {Near-MDS} or {2-MDS} {S}elf-{D}ual {C}odes via {T}wisted
  {G}eneralized {R}eed-{S}olomon {C}odes.
\newblock {\em IEEE Transactions on Information Theory 68}, 12 (2022),
  7832--7841.

\bibitem{NMDS_code_2022_2}
{\sc Sui, J., Zhu, X., and Shi, X.}
\newblock {MDS} and near-{MDS} codes via twisted {R}eed--{S}olomon codes.
\newblock {\em Designs, Codes and Cryptography 90}, 8 (Aug 2022), 1937--1958.

\bibitem{MUGI}
{\sc Watanabe, D., Furuya, S., Yoshida, H., Takaragi, K., and Preneel, B.}
\newblock A {N}ew {K}eystream {G}enerator {MUGI}.
\newblock In {\em Fast Software Encryption\/} (Berlin, Heidelberg, 2002),
  J.~Daemen and V.~Rijmen, Eds., Springer Berlin Heidelberg, pp.~179--194.

\end{thebibliography}

\end{document}